\documentclass[final]{siamltex}
\usepackage{url}
\usepackage{color}
\usepackage{algorithmic}
\usepackage[dvips]{graphicx}
\usepackage{epsfig}

\def\blue#1{{#1}}
\def\green#1{{#1}}

\xdef\old#1{}

\def\ao{}
\def\aoc{}

\def\jrev#1{{#1}}
\def\alex#1{{#1}}

\providecommand{\com[2]}{\begin{tt}[#1: #2]\end{tt}}
\providecommand{\prt}[1]{\left( #1 \right)}
\providecommand{\abs[1]}{\left|#1\right|}
\providecommand{\norm[1]}{\abs{\abs{#1}}}%

\begin{document}

\title{Matrix $p$-norms are NP-hard to approximate if $p \neq 1,
2, \infty$.}

\author{Julien M. Hendrickx \blue{and} Alex Olshevsky \thanks{The authors are with the
Laboratory for Information and Decision Systems, Massachusetts
Institute of Technology, Cambridge, MA, jm\_hend@mit.edu,
alex\_o@mit.edu}. This research was supported by the National
Science Foundation under grant ECCS-0701623. \green{Julien Hendrickx
holds postdoctoral fellowships from the F.R.S.-FNRS (Belgian
National Fund for Scientific Research) and the Belgian American
Education Foundation, and is on leave from the department of
mathematical engineering of the Universit\'e catholique de Louvain,
Louvain-la-Neuve, Belgium.}}

\maketitle

\begin{abstract} We show that for any rational $p \in [1,\infty)$ except
$p=1,2$, unless $P=NP$, there is no polynomial-time algorithm
\jrev{which approximates} the matrix $p$-norm to arbitrary relative
precision. We also show that for any rational $p \in [1,\infty)$
including $p=1,2$, unless $P=NP$, there is no polynomial-time
algorithm \jrev{which} approximates the $\infty,p$ mixed norm to
some fixed relative precision.
\end{abstract}

\section{Introduction}

{\aoc The $p$-norm of a matrix $A$ is defined as
\[ ||A||_p = \max_{\alex{||x||_p = 1}} ||Ax||_p.
\]} We consider the problem of computing the matrix $p$-norm \alex{to relative error $\epsilon$}, defined as
follows: given the inputs (i) a matrix $A \in R^{n \times n}$ with
rational entries (ii) an error tolerance $\epsilon$ which is a
positive rational number, output a rational number $r$ satisfying
{\aoc \[ \big|r - ||A||_p \big|  \leq \epsilon ||A||_p
\]} We will use the standard bit model of computation.  When $p=\infty$ {\aoc or $p=1$} the $p$-matrix norm is the largest of the
row/column sums, and thus may be {\aoc easily} computed exactly.
When $p=2$, this problem reduces to computing an eigenvalue of $A^T
A$ and thus can be solved in polynomial time in $n, \log
\frac{1}{\epsilon}$ and the bit-size of the entries of $A$. Our main
result suggests that the case of $p \notin \{1,2, \alex{\infty}\}$
may be different:

\bigskip

\begin{theorem} \label{mainthm}
For any rational $p \in [1,\infty)$ except $p=1,2$, unless $P=NP$,
there is no algorithm which \alex{computes the $p$-norm of a matrix
with entries in $\{-1,0,1\}$ to relative error $\epsilon$} with
running time polynomial in \alex{$n$}, $\frac{1}{\epsilon}$.
\end{theorem}

\bigskip

\green{{\aoc On the way to our result, we also slightly improve the
NP-hardness
 result for the mixed norm $||A||_{\infty,p}=\max_{||x||_\infty
\leq 1} ||Ax||_p$ from \cite{S05}. Specifically, we show that for
every rational $p {\aoc \geq } 1$, there exists an error tolerance
$\epsilon(p)$ such that unless $P=NP$, there is no polynomial time
algorithm approximating $||A||_{\infty,p}$ with a relative error
smaller than $\epsilon(p)$.}}

\subsection{Previous work}

When $p$ is an integer, computing the matrix norm can be recast as
solving a polynomial optimization problem. {\ao These are known to
be hard to solve in general \cite{MK87};  however, because the
matrix norm problem has a special structure one cannot immediately
rule out the possibility of a polynomial-time solution.} \blue{A few
hardness results are available in the literature} for mixed matrix
norms $||A||_{p,q}=\max_{||x||_p \leq 1} ||Ax||_q$. \blue{Rohn has
shown in} \cite{R00} that computing the $||A||_{\infty,1}$ norm is
NP-hard. \blue{In her thesis, Steinberg} \cite{S05} proved
\blue{more generally} that computing $||A||_{p,q}$ is NP-hard when
$1 \leq q<p \leq \infty$. We refer the reader to \cite{S05} for a
discussion of applications of the mixed matrix norm problems to
robust optimization.

It is conjectured in \cite{S05} that there are only three cases in
which mixed norms are computable in polynomial time: $p=1$ or $q =
\infty$ or $p=q=2$. Our work makes  progress on this question by
settling the ``diagonal'' case of $p=q$; however, the case of $p <
q$, as far as the authors are aware, is open.

\subsection{Outline}

\blue{We begin in Section \ref{infinitysubsection} by providing a
proof of the NP-hardness of \green{approximating} the mixed norm
$||\cdot||_{\infty,p}$ \green{within some fixed relative error,} for
any rational $p \geq 1$. {\aoc The proof may be summarized as
follows: \alex{observe that for any matrix M,} $\max_{||x||_\infty =
1}||Mx||_{p}$ is always attained \jrev{at} one \ao{of} the $2^n$
points of $\{-1,1\}^n$; so by appropriately choosing $M$, one can
encode the NP-hard problem of maximization over the latter set.}}
\alex{This argument will prove that computing the
$||\cdot||_{\infty,p}$ norm is NP-hard.}

\alex{Next, in Section \ref{sec:matrix_expo}} \blue{we exhibit a
class of matrices $A$ such that $\max_{||x||_p = 1}||A x||_{p}$ is
attained at \ao{each} of the $2^n$ points of $\{-1,1\}^n$ (up to
 scaling) and nowhere else.} \blue{These two elements are
combined in Section \ref{sec:proof_combine} to prove Theorem
\ref{mainthm}. More precisely, we define the matrix $Z = (M^T~
\alpha A^T)^T$, {\ao where we will pick $\alpha$ to be a large
number {\aoc depending on $n,p$} ensuring that the } maximum of
$||Zx||_p/||x||_p$ occurs very close to vectors $x\in \{-1,1\}^n$.
{\ao As mentioned several sentences ago, \alex{the value of
$||Ax||_p$ is the same for every vector $x \in \{-1,1\}^n$;}}  as a
result, the maximum of $||Zx||_p/||x||_p$ is determined by the
maximum of $||Mx||_p$ on \alex{$\{-1,1\}^n$}, which is proved in
Section \ref{infinitysubsection} to be hard to compute.} {\ao We
conclude with some remarks on the proof in Section \ref{sec:ccl}.}

\section{The $||\cdot||_{\infty,p}$ norm}\label{infinitysubsection}

We \blue{now describe} a simple construction which relates the
$\infty,p$ norm to the maximum cut in a graph.

Suppose $G = (\{1,\ldots,n\},E)$ is an undirected, \alex{connected}
graph. We will use $M(G)$ to denote the edge-vertex incidence matrix
of $G$; that is, $M(G) \in R^{|E| \times n}$; we will think of
columns of $M(G)$ as corresponding to nodes of $G$ and rows of
$M(G)$ as corresponding to the edges of $G$. The entries of $M(G)$
are as follows: orient the edges of $G$ arbitrarily, and let the
$i$'th row of $M\jrev{(G)}$ have $+1$ in the column corresponding to
the origin of the $i$'th edge, $-1$ in the column corresponding to
the endpoint of the $i$'th edge, and $0$ at all other columns.

Given any partition of $\{1,\ldots,n\} = S \cup S^c$, we define
${\rm cut}(G,S)$ to be the number of edges with exactly one endpoint
in $S$. Furthermore, we define ${\rm maxcut}(G) = \max_{S \subset
\{1,\ldots,n\}} {\rm cut}(G,S).$ The indicator vector of a cut $(S,
S^c)$ is the vector $x$ with $x_i=1$ when $i \in S$ and $x_i = -1$
when $i \in S^c$. We will use ${\rm cut}(x)$ for vectors $x \in
\{-1,1\}^n$ to denote the value of the cut whose indicator vector is
$x$.

\bigskip

\begin{proposition} \label{cutbound}\jrev{For any $p\geq 1$,} \[ \max_{||x||_{\infty} \leq 1}
||M\jrev{(G)} x||_p = 2{\rm maxcut}(G)^{1/p}.\]
\end{proposition}

\bigskip

\begin{proof} Observe that $\norm{M\jrev{(G)}x}_p$ is a convex function of $x$,
so that the maximum is achieved at the extreme points of the set
$||x||_{\infty} \leq 1$, i.e. vectors $x$ satisfying $x_i = \pm 1$.
Suppose we are given such a vector $x$; define $S = \{ i ~|~
x_i=1\}$. Clearly, $\norm{M\jrev{(G)}x}_p^p = 2^p {\rm cut}(G,S)$.
From this the proposition immediately follows.
\end{proof}

\bigskip

Next, we introduce an error term into this proposition. Define $f^*$
to be the optimal value $f^* = \max_{||x||_{\infty} \leq 1}
\norm{M\jrev{(G)}x}_p$; the above proposition implies that
$(f^*/2)^p = {\rm maxcut}(G)$. We want to argue that if $f_{\rm
approx}$ is close enough to $f^*$, then $(f_{\rm approx}/2)^p$ is
close to ${\rm maxcut}(G)$.

\bigskip

\begin{proposition}\label{prop:relative_error}
\green{If \alex{$p \geq 1$}, $|f^* - f_{\rm approx}| < \epsilon f^*$
with $\epsilon <1$, then
$$
\abs{\prt{\frac{f_{\rm approx}}{2}}^p- {\rm maxcut}(G)} \leq
2^{p-1}p\epsilon \cdot {\rm maxcut}(G)
$$}
\end{proposition}
\bigskip

\begin{proof}
\green{By Proposition \ref{cutbound} ${\rm maxcut}(G)  =
(f^*/2)^p$. Using the inequality
\[ |a^p - b^p| \leq |a-b| p \max(|a|,|b|)^{p-1},  \]
we obtain
$$
\abs{\prt{\frac{f_{\rm approx}}{2}}^p- {\rm maxcut}(G)} \leq
\frac{1}{2}\abs{f^* - f_{\rm approx}} p \max\prt{\frac{f^*}{2}
,\frac{f_{\rm approx}}{2} }^{p-1}.
$$
It follows from $\epsilon <1$ that $f_{\rm approx} \leq 2 f^*$. We
have therefore
$$
\abs{\prt{\frac{f_{\rm approx}}{2}}^p- {\rm maxcut}(G)} \leq
\frac{1}{2}\abs{f^* - f_{\rm approx}} \cdot p \cdot (f^*)^{p-1} \leq
\frac{\epsilon}{2}p(f^*)^{p},
$$
where we have used the assumption that $|f^* - f_{\rm approx}|\leq
\epsilon f^*$. The result follows then from  ${\rm maxcut}(G)  =
(f^*/2)^p$.}
\end{proof}

\bigskip

We now put together the previous two propositions to prove that
\green{approximating the $||\cdot||_{\infty,p}$ norm within some
fixed relative error is NP-hard}.

\bigskip

\begin{theorem} \label{inftykthm}
For any rational $p \geq 1$, and $\delta>0$, unless $P=NP$, there is
no algorithm \alex{which given a matrix with entries in $\{-1,0,1\}$ computes its
$p$-norm to relative error $\epsilon= \prt{(33+\delta)p2^{p-1}}^{-1}$ with
running time polynomial in the dimensions of the matrix.}
\end{theorem}

\bigskip

\begin{proof}
\green{Suppose there was such an algorithm. Call $f^*$ its output on
the \alex{$|E| \times n$} matrix $M(G)$ for a given \alex{connected} graph $G$ on $n$ vertices.
It follows from Proposition \ref{prop:relative_error} that
$$ \abs{\prt{\frac{f_{\rm
approx}}{2}}^p- {\rm maxcut}(G)} \leq
\frac{2^{p-1}p}{(33+\delta)p2^{p-1}} {\rm maxcut}(G) =
\frac{1}{33+\delta}{\rm maxcut}(G). $$
Observing that
$$
\frac{32+\delta}{34+\delta} {\rm maxcut}(G) =
\frac{33+\delta}{34+\delta} \prt{{\rm maxcut}(G) -
\frac{1}{33+\delta}{\rm maxcut}(G)},
$$
the former inequality implies
$$\frac{32+\delta}{34+\delta} {\rm maxcut}(G) \leq
\frac{33+\delta}{34+\delta}\prt{\frac{f_{\rm approx}}{2}}^p \leq
{\rm maxcut}(G).
$$
Since $p$ is rational, one can compute in polynomial time a lower
bound $V$ for $\frac{33+\delta}{34+\delta}(f_{\rm approx}/2)^p$
sufficiently accurate so that $V> \frac{32+\delta/2}{34+\delta/2}
{\rm maxcut}(G)>\frac{16}{17}{\rm maxcut}(G)$. However, it has been
established in \cite{H01} that unless $P=NP$, for any $\delta'>0$,
there is no algorithm producing a quantity $V$ in polynomial time in
$n$ such that
$$
\prt{\frac{16}{17}+\delta'}{\rm maxcut}(G)\leq V \leq {\rm
maxcut}(G).
$$}
\hspace{12 cm}
\end{proof}

\bigskip

\alex{\noindent {\bf Remark:} Observe that the matrix $M(G)$ is not square. If one desires
to prove hardness of computing the $\infty,p$-norm for square matrices, one can simply
add $|E|-n$ zeros to every row of $M(G)$. The resulting matrix has the same $\infty,p$ norm
as $M(G)$, is square, and its dimensions are at most $n^2 \times n^2$.}

\section{\blue{A discrete set of exponential size}}\label{sec:matrix_expo}

\blue{Let us now fix $n$ and a rational $p>2$. We denote by $X$ the
set $\{-1,1\}^n$, and use $S(a,r)=\{ x \in R^n ~|~ ||x-a||_p = r \}$
to stand for the sphere of radius $r$  around $a$ in the $p$-norm.
We consider the } following matrix in $R^{2n \times n}$:
\[A = \left(
\begin{array}{rrrrr}
                                                           1 & -1 &  &  &
 \\
                                                           1 & 1 &  &  &
                                                           \\ \hline
                                                            & 1 & -1 &  &
 \\
                                                            & 1 & 1 &  &
                                                            \\
                                                            \hline
                                                            &  & \ddots &
\ddots &  \\

                                           &  & \ddots & \ddots &

                                           \\

                                           \hline
                                                            &  &  & 1 & -1
\\
                                                            &  &  & 1 & 1
                                                            \\
                                                            \hline
                                                           -1 &  & &  & 1 \\
                                                           1 &  & &  & 1 \\
                                                         \end{array}
                                                       \right).
 \]
\blue{and show that the maximum of $||Ax||_p$ for $x\in
S(0,n^{1/p})$ is attained \jrev{at} the $2^n$ vectors in $X$ and no
other points.} {\ao For this, we will need the following lemma.}

\bigskip

\begin{lemma}
\label{basicineq} For any real numbers $x,y$ and $p \geq 2$
\[ |x+y|^p + |x-y|^p \leq 2^{p-1} \prt{|x|^p + |y|^p}. \] In fact,
$|x+y|^p + |x-y|^p $ is upper bounded by
 \[ 2^{p-1} \prt{ |x|^p + |y|^p} - \frac{\prt{|x|-|y|}^2}{4}
\prt{ p(p-1)\big||x|+|y|\big|^{p-2} - 2 \big||x|-|y|\big|^{p-2}},
\] \blue{where the last term on the right is always
nonnegative.}
\end{lemma}
\begin{proof}
By symmetry we can assume that $x \geq y \geq 0$. In that case, we
need to prove
\[ (x+y)^p + (x-y)^p \leq 2^{p-1}
(x^p + y^p)- \frac{(x-y)^2}{4} \prt{p (p-1) (x+y)^{p-2} - 2
(x-y)^{p-2}}.\] Divide both sides by $(x+y)^p$ and change variables
to $z=(x-y)/(x+y)$:
\[  1+z^p \leq \frac{(1+z)^p + (1-z)^p}{2} - \prt{\frac{p(p-1)}{4}  z^2 -
\frac{1}{2} z^p}.  \] The original inequality holds if this
inequality holds for $z \in [0,1]$. \jrev{Let's} simplify:
\[ 2 + z^p \leq (1+z)^p + (1-z)^p - \frac{p (p-1)}{2} z^2. \]
Observe that we have equality when $z=0$, so it suffices to show the
right-hand side grows faster than the left-hand side, namely:
\[ z^{p-1} \leq (1+z)^{p-1} - (1-z)^{p-1} - (p-1) z,\] and this
\blue{follows from}
$$ (1+z)^{p-1}  \geq  1 + (p-1) z  \geq   (1-z)^{p-1} +
z^{p-1} + (p-1) z,
$$
\blue{where we have used the convexity of $f(a)=a^{p-1}$.}
\end{proof}

\bigskip

Now we are ready to \blue{prove} that vectors in $X$ optimize
$||Ax||_p/||x||_p$, or, equivalently, optimize \blue{$||Ax||_p^p$
over the sphere $S(0,n^{1/p})$}.

\bigskip

\begin{lemma} \label{upperboundlemma} For any $p \geq 2$, the supremum of
\blue{$||Ax||_p^p$} over $S(0,n^{1/p})$ is achieved by any vector in
$X$.
\end{lemma}

\begin{proof} Observe
that $\blue{||Ax||_p^p} = n2^{p}$ for any $x \in X$. \blue{To prove
that this is the largest possible value, we write}
\begin{equation} \label{gdef} \norm{Ax}_p^p = \sum_{i=1}^n |x_i - x_{i+1}|^p + |x_i+ x_{i+1}|^p,
\end{equation}
\blue{using the convention $n+1 = 1$ for the indices.} Lemma
\ref{basicineq} implies that
\[  |x_i - x_{i+1}|^p + |x_i+ x_{i+1 }|^p
\leq 2^{p-1} \prt{|x_i|^p + |x_{i+1}|^p}.\]
\blue{Applying this inequality to each therm of (\ref{gdef}) and
using $||x||_p^p = n$, we obtain}
$$ \norm{Ax}_p^p \leq  \sum_{i=1}^n 2^{p-1}\prt{|x_i|^p + |x_{i+1}|^p} = 2^p\sum_{i=1}^n |x_i|^p =
2^pn.$$\hspace{12cm}
\end{proof}

\bigskip

Next, we refine the previous lemma by including a bound on how fast
$\blue{\norm{Ax}_p^p}$ decreases as we move a little bit away from
the set $X$\blue{ while staying on $S(0,n^{1/p})$}.

\bigskip

\begin{lemma} \label{deficiencybound} Let $p \geq 2, c \in \blue{(0,1/2}]$ and
suppose $y \in S(0,n^{1/p})$ has the property that
\begin{equation} \label{lowerbound1} \min_{x \in X} ||y - x||_{\infty} \geq
c.
\end{equation} Then, \[ \blue{\norm{Ay}_p^p} \leq \blue{n}2^p - \frac{3(p-2)}{\alex{2^p}n^{\blue{2}}} c^2.\]
\end{lemma}
\begin{proof} We proceed as before in the proof of Lemma
\ref{upperboundlemma}, until the time comes to apply Lemma
\ref{basicineq}, when we include the error term which we had
previously ignored:

\[ \blue{\norm{Ay}_p^p}  \leq  n2^p -
\frac{1}{4}  \sum_i \prt{|y_i| - |y_{i+1}|}^2 \prt{p(p-1)~
\big||y_i|+|y_{i+1}|\big|^{p-2} - 2~
\big||y_i|-|y_{i+1}|\big|^{p-2}},
\] {\ao Note that on the right-hand side, we are subtracting a sum
of nonnegative terms. The upper bound will still hold if we subtract
only one of these terms; so we conclude that for each $k$,} \[ {\ao
||Ay||_p^p} \leq  n2^p - \blue{\frac{1}{4}} \prt{|y_k| - |y_{k+1
}|}^2 \prt{p(p-1)~ \big||y_k|+|y_{k+1}|\big|^{p-2} - 2 ~
\big||y_i|-|y_{i+1}|\big|^{p-2}}. \]

\jrev{By assumption, there is at least one $y_k$ with
$\big||y_k|-1\big| \geq c$. Suppose first that $\abs{y_k}>1$. Then
we have $\abs{y_k}>1 +c$, and there must be an $y_j$ with
$\abs{y_j}<1$ for otherwise $y$ would not be in $S(0,n^{1/p})$.
Similarly, if $\abs{y_k}<1$, then $\abs{y_k}< 1-c$ and there is a
$j$ for which $\abs{y_j}>1$. In both cases, this implies the
existence of an index $m$ with $|y_m|$ and $|y_{m+1}|$ differing
by at least $c/n$ and such that at least one of $|y_m|$ and
$\abs{y_{m+1}}$ is larger than or equal to $1-c$. Therefore,}
%
\[ \blue{\norm{Ay}_p^p} \leq \blue{n}2^p - \blue{\frac{1}{4}} \frac{c^2}{n^2} \left[ p(p-1) ~ \big||y_{\blue{m}}| +
|y_{\blue{m+1}}|\big|^{p-2} - 2 ~ \big||y_{\blue{m}}| -
|y_{\blue{m+1}}|\big|^{p-2}\right].
\]
Now observe that $ \big||y_{\blue{m}}|-|y_{\blue{m}+1}|\big|\leq
|y_{\blue{m}}| + |y_{\blue{m+1}}|$, \blue{and that $|y_m| +
|y_{m+1}| \geq (1-c) \geq 1/2$ because $c \in (0,1/\blue{2}]$}.
\alex{These two inequalities suffice to establish that t}he term in
square brackets is at least $(1/2)^{\alex{p-2}} (p(p-1)-2) \geq
(3/2^{\alex{p}})(p-2)$, so that
\[  \blue{\norm{Ay}_p^p} \leq \blue{n}2^p - \frac{3(p-2)}{\alex{2^p} n^\blue{2}} c^2. \]
\hspace{12cm}
\end{proof}

\section{\blue{Proof of Theorem \ref{mainthm}}}\label{sec:proof_combine}

\blue{We now relate the results of the last two sections to the
problem of the $p$-norm}. For a suitably defined matrix $Z$
\blue{combining $A$ and $M(G)$}, we want to argue that the optimizer
of $||Zx||_p/||x||_p$ is very close to satisfying $|x_i|=|x_j|$
\jrev{for every $i,j$}.

\bigskip

\begin{proposition} \label{extremabound} Let $p>2$,  \blue{and
$G$ a graph on $n$ vertices. Consider the matrix}
\[ \tilde Z = \left(\begin{array}{c}A \\\frac{p-2}{64 p n^{\blue{8}}}M(G)\\
\end{array}
\right), \] \blue{with $M(G)$ and $A$ as in Sections
\ref{infinitysubsection} and \ref{sec:matrix_expo} respectively.} If
$x^*$ is the vector at which the optimization problem $ \max_{x \in
S(0,n^{1/p})} ||\tilde Zx||_p $ achieves its supremum then
\[ \min_{x \in X} ||x^* - x||_{\infty} \leq \frac{1}{4^p n^{\blue{6}}}.\]
\end{proposition}

\begin{proof} Suppose the conclusion is false; then using Lemma
\ref{deficiencybound} with $ c = 1/4^p n^{\blue{6}}$, \blue{we
obtain}
\[ ||A x^*||_p^p   \leq n 2^p  -\frac{3(p-2)}{\alex{2^p}4^{2p}
n^{\blue{14}}} \alex{ = n 2^p  -\frac{3(p-2)}{32^p
n^{\blue{14}}}}.\] \blue{It follows from} Proposition \ref{cutbound}
\blue{that}
\[ ||Mx^*||_p^p \leq 2^p {\rm maxcut}(G) \leq 2^p n^2,\]
so that
$$ ||\tilde Zx^*||_p^p  =  ||Ax^*||_p^p +
\prt{\frac{p-2}{64 p n^{\blue{8}}}}^p ||Mx^*||_p^p
  \leq   2^p n - \frac{3(p-2)}{\alex{32^p} n^{\blue{14}}} +
\frac{ 2^p(p-2)^p  n^2}{ 64^p  p^p n^{\blue{8}p}}.
$$
\blue{Observe that the last term in this inequality is smaller than
the previous one (in absolute value). Indeed, for $p>2$, we have
that  $3/\alex{32^p}>(2/64)^p$, $p-2> [(p-2)/p]^p$ and $1/n^{14} >
n^2/n^{8p}$. We therefore have $||Zx^*||_p^p < 2^pn$.} By contrast,
let $x$ be any vector in $\{-1,1\}^n$. Then $x \in S(0, n^{1/p})$
and
\[ ||\tilde Zx||_p^p \geq ||Ax||_p^p \geq 2^p n,\] which contradicts the optimality
of $x^*$.
\end{proof}

\bigskip

Next, we seek to translate the fact that the optimizer $x^*$ is
close to $X$ to the fact that the objective value $||Zx||_p/||x||_p$
is close to the largest objective value at $X$.

\bigskip

\begin{proposition} \label{objectivebound}
Let $p>2$,  \blue{$G$ a graph on $n$ vertices, and}
\[ Z = \left(
                                                 \begin{array}{c}
                                                   \frac{64 p n^{\blue{8}}}{p-2} A
\\                                                    M(G) \\
                                                 \end{array}
                                               \right). \]
If $x^*$ is the vector at which the optimization problem
\[ \max_{x \in S(0,n^{1/p})} ||Zx||_p \] achieves its supremum and $x_{\rm
r}$ \alex{is} the rounded version of $x^*$ in which every component
is rounded to the closest \jrev{of $-1$ and $1$}, then
\[ \Big| ~ ||Z x^*||_p^p - ||Z x_{\rm r}||_p^p\Big| \leq
\frac{1}{n^{2}}.
\]
\end{proposition}
\begin{proof} Observe that $x^*$ is the same as the extremizer of the
corresponding problem with $\tilde Z$ instead of $Z$, so that $x$
satisfies the conclusion of Proposition \ref{extremabound}.
\blue{Consequently every component of $x^*$ is closer to one of $\pm
1$ than to the other, and so $x_{\rm r}$ is well defined}. We have:
\[ ||Zx^*||_p^p - ||Z x_{\rm r}||_p^p = \prt{64 \frac{p}{p-2} n^{\blue{8}}}^p (||A
x^*||_p^p - ||A x_{\rm r}||_p^p) + (||M x^*||_p^p - ||M x_{\rm
r}||_p^p).\] This entire quantity is nonnegative since $x^*$ is the
maximum of $||Zx||$ on $S(0, n^{1/p})$.  Moreover, $||A x^*||_p^p -
||A x_{\rm r}||_p^p$ is nonpositive, since by Proposition
\ref{upperboundlemma} $||Ax||_p$ achieves its maximum over $S(0,
n^{1/p})$ on all the elements of $X$. Consequently,
\begin{eqnarray} || \jrev{Z} x^*||_p^p - ||\jrev{Z} x_{\rm r}||_p^p & \leq &
||M
x^*||_p^p - ||M x_{\rm r}||_p^p \nonumber \\
& \leq &  (||Mx^*||_p - ||M x_{\rm r}||_p) p \max( ||M x^*||_p, ||M
x_{\rm r}||_p)^{p-1}. \label{equationtobound}
\end{eqnarray}
\blue{We now bound} all the terms in the last equation. First,
\begin{equation}\label{eq:Delta||Mx||}
||M x^*||_p - ||M x_{\rm r}||_p \leq ||M||_2 || x^*- x_{\rm r} ||_2
\leq ||M||_F \sqrt{n} || x^*- x_{\rm r}||_{\infty} = \frac{n
\sqrt{n}}{4^p n^{6}},
\end{equation}
\blue{where we have used $||M(G)||_F = \sqrt{2\abs{E}}< \green{n}$
and Proposition \ref{extremabound} for the last inequality. Now that
we have a bound on the first term in Eq. (\ref{equationtobound}), we
proceed to the last term.}  \blue{It follows from the definition of
$M$ that}
\[ ||Mx _{\rm r}||_p^p \leq 2^p \cdot {n \choose 2} \leq
2^p n^2.\]  \alex{Next, we bound $||Mx^*||_p^p$.  Observe that a
particular case of Eq. (\ref{eq:Delta||Mx||}) is: \begin{equation}
\label{plusone} ||Mx^*||_p < ||M x_{\rm r}||_p +1.\end{equation}
Moreover, observe that $||Mx_{\rm r}||_p \geq 1$ (the only way
this does not hold is if every entry of $x_{\rm r}$ is the same,
i.e. $||Mx_{\rm r}||_p=0$; but then Eq. (\ref{plusone}) implies
that $||Mx^*||_p < 1$, which is impossible since $G$ has at least
one edge), and so Eq. (\ref{plusone}) implies $||Mx^*||_p \leq 2
||M x_{\rm r}||_p$, and so:}
\[ ||M x^*||_p^p \leq 4^{p} n^2.\]
\alex{Thus} $\max( ||M x^*||_p, ||M x_{\rm r}||_p)^{p} \alex{\leq
4}^pn^2$ and thereforet $\max( ||M x^*||_p, ||M x_{\rm r}||_p)^{p-1}
\alex{ \leq 4}^p n^2$. Indeed, this bound is trivially valid if
$\max( ||M x^*||_p, ||M x_{\rm r}||_p)^{p}\leq 1$, and follows from
$a^{p-1} < a^p$ for $a\geq 1$ otherwise. Using this bound and the
inequality (\ref{eq:Delta||Mx||}), we finally obtain
\[ || Z
 x^*||_p^p - ||Z x_{\rm r}||_p^p \leq \frac{n^{1.5}}{4^p n^{6}} p\cdot \alex{4}^{p}
n^2 \leq \frac{1}{n^{2}} . \] \hspace{12cm}
\end{proof}

\bigskip

Finally, let us bring it all together by arguing that if we can
approximately compute the $p$-norm of $Z$, we can approximately
compute the maximum cut.

\bigskip
\begin{proposition} Let $p>2$. Consider \blue{a graph $G$ on \alex{$n>2$} vertices and} the matrix \[ Z = \left(
                                                 \begin{array}{c}
                                                   64 \frac{p}{p-2} n^{8} A
\\
                                                    M(G) \\
                                                 \end{array}
                                               \right), \] \blue{and let } $f^* = ||Z||_p$. If
\begin{equation} \label{fapprox} |f_{\rm approx} - f^*| \leq
\frac{(p-2)^p}{{\ao \green{132}}^p p^p n^{8p+\green{3}} p},
\end{equation}
then
\[ \abs{\prt{\frac{ n}{2^p} f_{\rm approx}^p -
n\prt{\frac{64pn^{8}}{p-2}}^p} - {\rm maxcut}(G)} \leq \frac{1}{n}.
\] \label{cutapprox}
\end{proposition}

\bigskip
\begin{proof}
Observe that $ n^{\frac{1}{p}} f^{*} = \max_{x \in S(0,n^{1/p})}
||Zx||_p$. It follows thus from Proposition \ref{objectivebound}
that
\begin{equation}\label{eq:boundnf*}
 \abs{nf^{*p} - \max_{x
\in X} ||Zx||_p^p }<\frac{1}{n^2}.
\end{equation} Recall that
$||Zx||_p^p = ||Mx||^p_p + \prt{64 \frac{p}{p-2} n^{8}}^p ||Ax||_p^p
$, and that $||Ax||_p^p= n2^p$ for every $x\in X$. Therefore,
$$
\max_{x \in X} ||Zx||_p^p = \prt{\frac{64pn^{8}}{p-2} }^pn2^p +
\max_{x \in X} ||Mx||_p^p = \prt{\frac{64pn^{8}}{p-2} }^pn2^p + 2^p
{\rm maxcut}(G),
$$
and {\ao combining the last two equations} we have
\begin{equation} \abs{\prt{\frac{ n}{2^p} f^{*p} - n\prt{\frac{64pn^{8}}{p-2}}^p} - {\rm
maxcut}(G)} \leq \frac{1}{2^p n^{2}}, \label{fstareq} \end{equation}
Let us now evaluate the error introduced by the approximation
$f_{\rm approx}$.
\begin{eqnarray*}
\abs{\prt{\frac{ n}{2^p} f_{\rm approx}^p -
n\prt{\frac{64pn^{8}}{p-2}}^p} - {\rm maxcut}(G)} &\leq&
\frac{1}{2^p n^{2}} + \frac{n}{2^p}\abs{f_{\rm approx}^p - f^{*p}}
\\
&\leq&\frac{1}{2^p n^{2}} + \frac{n}{2^p}\abs{f_{\rm approx} -
f^{*}} p \max( f^*, f_{\rm approx})^{p-1}.
\end{eqnarray*}
It remains to bound the last term of  this inequality. {\ao First,
we use the fact that $f^* \geq 1$ and  Eq. (\ref{fstareq}) to argue}
\begin{equation} \label{pnormupperbound}
f^{*(p-1)}\leq f^{* p}\leq 2^p\prt{\frac{64pn^{8}}{p-2}}^p +
\frac{2^p}{n} {\rm maxcut}(G) + \frac{1}{n^3} \leq 2^p
\prt{\frac{66pn^{8}}{p-2}}^p,
\end{equation}
where we have used ${\rm maxcut}(G) < n^2$ and $1\leq p/(p-2)$ for
the last inequality. By assumption, $|f_{\rm approx} - f^*|\leq 1$
and since $f^* \geq 1$,
$$
f_{\rm approx}^{(p-1)}\leq ({\ao 2 f^{*})^{p-1} \leq  (2
f^{*})^{p}\leq 4^p \prt{\frac{66pn^{8}}{p-2}}^p.}
$$
Putting it all together and using the bound on $|f_{\rm approx} -
f^*|$, we obtain (assuming $n>1$)
\begin{eqnarray*}
\abs{\prt{\frac{ n}{2^p} f_{\rm approx}^p -
n\prt{\frac{64pn^{8}}{p-2}}^p} - {\rm maxcut}(G)} &\leq&
\frac{1}{2^p n^{2}} + {\ao \frac{(p-2)^p}{\green{132}^p p^p n^{8
p+\green{3}} p} \blue{2}^p n p \prt{\frac{66pn^{8}}{p-2}}^p } \\
&\leq& \frac{1}{2^p n^{2}} +
\frac{1}{n^2}\\
&\leq& \frac{1}{n}.
\end{eqnarray*}\hspace{12 cm}
\end{proof}

\bigskip

\begin{proposition} \label{mainproposition}
Fix \blue{a rational} $p \in [1,\infty)$ with $p \neq 1,2$. Unless
$P=NP$, there is no algorithm which, given input $\epsilon > 0$ and
a matrix $Z$, computes $||Z||_p$  \green{to a relative accuracy
$\epsilon$}, in time which is \alex{polynomial in $1/\epsilon$, the dimensions of $Z$,
 and the bit-size of the entries of $Z$.}
\end{proposition}
\bigskip

\begin{proof}
\green{Suppose first that $p>2$. We show that such an algorithm
could be used to build a polynomial-time algorithm solving the
maximum cut problem. \jrev{For} a graph $G$ on $n$ vertices, fix
$$\epsilon =   \prt{132^p\prt{\frac{p}{p-2}}^pn^{8p+\green{3}}p}^{-1} \cdot \prt{\alex{132} \prt{\frac{p}{p-2}}n^{8}}^{-1} ,$$
build the matrix $Z$ as in Proposition \ref{cutapprox}, and compute
the norm of $Z$; let $f_{\rm approx}$ be the output of the
algorithm. Observe that \alex{By Eq. (\ref{pnormupperbound})
$$
\alex{||Z||_p} \leq \frac{132 p n^8}{p-2},
$$ so,}
$$
\abs{f_{\rm approx}- \norm{Z}_p}\leq \epsilon \norm{Z}_p \leq
\epsilon \prt{\alex{132}\frac{p}{p-2}n^8} \leq
\prt{132^p\prt{\frac{p}{p-2}}^pn^{8p+3}p}^{-1}
$$
It follows then from Proposition \ref{cutapprox} that}
\[  n \prt{\frac{f_{\rm approx}}{2}}^p - n\prt{64 \cdot \prt{\frac{p}{p-2}}
n^{8}}^p  \] \blue{is an approximation of the maximum cut with an
additive error at most $1/n$. {\ao Once we have  $f_{\rm approx}$}},
we can approximate this number in polynomial time to an
\alex{additive} accuracy of $1/4$.  This gives an additive error
$1/4+1/n$ approximation algorithm for maximum cut, and since the
maximum cut \jrev{is} always an integer, this means we can compute
it exactly when $n>4$. However, maximum cut is an NP-hard problem
\cite{GJ79}.

For the case of $p \in (1,2)$, NP-hardness follows from the analysis
of the case of $p>2$ since \alex{for any matrix $Z$, $||Z||_p = ||Z^T||_{p'}$} where $1/p +
1/p'=1$.
\end{proof}

\bigskip

{\aoc \noindent {\bf Remark:} In contrast to Theorem \ref{inftykthm}
which proves the NP-hardness of computing the matrix $\infty,k$ norm
to relative accuracy $\epsilon = 1/C(p)$, for some function $C(p)$,
Proposition  \ref{mainproposition} proves the NP-hardness of
computing the $p$-norm to accuracy $1/C'(p) n^{8p+11}$, for some
function $C'(p)$. In the latter case, $\epsilon$ depends on $n$.  }

\bigskip

\alex{Our final theorem demonstrates that the $p$-norm is still hard
to compute when restricted to matrices with entries in
$\{-1,0,1\}$.}

\bigskip

{\aoc \begin{theorem} \label{mainthm_bis} Fix a rational $p \in
[1,\infty)$ with $p \neq 1,2$. Unless $P=NP$, there is no algorithm
which, given input $\epsilon$ and a matrix $M$ with entries in
$\{-1,0,1\}$, computes $||M||_p$ to relative accuracy
$\epsilon$, in time which is polynomial in $\epsilon^{-1}$ and the
dimensions of the matrix.
\end{theorem}
\bigskip
\begin{proof}  \alex{As before, it suffices to prove the theorem for the
case of $p>2$; the case of $p \in (1,2)$ follows because $||Z||_p = ||Z^T||_{p'}$ where $1/p+1/p'=1$.}

 Define
\[ Z^* = \left(
                                                 \begin{array}{c}
                                                   \Big(\Big\lceil \prt{64 \frac{p}{p-2} n^{8}} \Big\rceil\Big) A
\\
                                                    M(G) \\
                                                 \end{array}
                                               \right) \] where $\lceil \cdot \rceil$ refers to
                                               rounding up to the closest integer. Observe
that by an argument similar to the proof of the previous
proposition, computing $||Z^*||_p$ to an accuracy $\epsilon = (C(p)
n^{8p+11})^{-1}$ is NP-hard for some function $C(p)$. But if we
define
\[ Z^{**} = \left(
                                                    \begin{array}{c}
                                                      A \\
                                                      A \\
                                                      \vdots \\
                                                      A \\
                                                      M \\
                                                    \end{array}
                                                  \right)\] where $A$ is repeated \jrev{$\Big\lceil \prt{64 \frac{p}{p-2} n^{8}}^{p} \Big\rceil$} times,
\alex{then} \[ ||Z^{**}||_p = ||Z^*||_p.\] The matrix $Z^{**}$ has
entries in $\{-1,0,1\}$ and its size is polynomial in $n$, so it
follows
                                                  that it is NP-hard to compute $||Z^{**}||_p$ within the same $\epsilon$.
\end{proof}}

\bigskip

\alex{ \noindent {\bf Remark:} Observe that the argument also suffices to show that computing
the $p$-norm of {\em square} matrices with entries in $\{-1,0,1\}$ is NP-hard: simply pad
each row of $Z^{**}$ with enough zeros to make it square. Note that this trick was also
used in Section \ref{infinitysubsection}.}

\section{Concluding remarks} \label{sec:ccl}

We have proved the NP-hardness of computing the matrix $p$-norm
approximately with {\em \green{relative}} error {\aoc $\epsilon =
1/C(p) n^{8p+11}$, where $C(p)$ is some function of $p$; and the
NP-hardness of computing the matrix $\infty,p$ norm to  some fixed
relative accuracy depending on $p$.} {\aoc We finish with some
technical remarks about various possible extensions of the theorem:}

\bigskip

\begin{itemize}
\item \green{Due to the linear property of the norm $\norm{\alpha
A}= \abs{\alpha}\norm{A}$, {\aoc our results also imply} the
NP-hardness of approximating the matrix $p$-norm with any fixed or
polynomially growing {\em additive error}}.

\item {\aoc Our construction is also implies the hardness of
computing the matrix $p$-norm for any irrational number $p>1$ for
which a polynomial time algorithm to approximate $x^{p}$ is
available.}

\item{Our construction may also be used to \blue{provide a new
proof of} the NP-hardness of the $|| \cdot ||_{p,q}$ norm when $p
\alex{>} q$, \blue{which has been established in \cite{S05}}.
Indeed, \blue{it} rests on the matrix $A$ with the property that
$\max ||Ax||_p/||x||_p$ occurs at the vectors $x\in \{-1,1\}^n$. We
use this matrix $A$ to construct the matrix $Z=(\alpha A ~M)^T$ for
large $\alpha$, and argue that $\max ||Zx||_p/||x||_p$ occurs close
to the vectors $x\in \{-1,1\}^n$. At these vectors, it happens $Ax$
is a constant, so we are effectively maximizing $||Mx||_p$, which is
hard as shown in Section
\ref{infinitysubsection}.\\
If one could come up with such a matrix for the case of the mixed
$||\cdot||_{p,q}$ norm, one could prove NP-hardness by following the
same argument. However, when $p>q$, actually the very same matrix
$A$ works. Indeed, one could simply argue that
\[ \alex{||A||_{p,q} = } \max_{\alex{x \neq 0}} \frac{||Ax||_q}{||x||_p} = \max_{\alex{x \neq 0}}
\frac{||Ax||_q}{||x||_q} \frac{||x||_q}{||x||_p},\] and since the
maximum of $||x||_q/||x||_p$ when $1 \leq q < p \leq \infty$ occurs
at the vectors $x\in \{-1,1\}^n$, we have that both terms on the
right are maximized at $x=\in \{-1,1\}^n$, that is where
$||Ax||_q/||x||_p$ is maximized.}

\item{Finally, we note that our goal was only to show existence of
a polynomial-time reduction from the maximum cut problem to the
problem of matrix  $p$-norm computation.  It is possible that more
economical reductions which scale more gracefully with $n$ and $p$
exist.}
\end{itemize}

\end{document}